\documentclass[conference,a4paper]{IEEEtran}
\usepackage{cite}
\usepackage{amsmath,amssymb,amsfonts,amsthm}
\usepackage{graphicx}
\usepackage{booktabs}
\usepackage{xcolor}
\usepackage{balance}
\newtheorem{theorem}{Theorem}
\newcommand{\Real}{\mathcal{R}}
\newcommand{\Q}{Q}
\newcommand{\conv}{\operatorname{conv}}

\begin{document}
\title{Beyond No-Good Benders Cuts: Exact Realizability\newline
for Discretely Tunable Clock Trees}
\author{\IEEEauthorblockN{Boxiang Song and Keren Zhu}
\IEEEauthorblockA{College of Integrated Circuits and Micro-Nano Electronics\\
Fudan University, Shanghai, China}}
\maketitle

\begin{abstract}
Useful-skew schedules can satisfy timing constraints yet remain unrealizable
by a fixed clock tree with discrete tuning choices. We formulate this mismatch
as exact membership in a finite relative-latency set and develop a checkable
feedback interface between the scheduler and the tree model. Arithmetic
certificates explain unrealizable targets, while tree-specific contraction
reduces the structural size of the exact relations projected onto selected
sinks. These relations become realizability cuts that can exclude more candidates
than a no-good on the same certificate support, including discrete holes that
linear inequalities cannot separate. Controlled experiments confirm fewer
oracle calls and faster projection construction. They also expose important
limits: globally minimum certificates can cost more than they save, compact
arithmetic feedback can fail on non-parity obstructions, and direct monolithic
optimization remains faster on the tested additive models. The contribution
is an exact, independently verifiable scheduler--tree interface, rather than a
claim of universal solver acceleration or physical signoff.
\end{abstract}

\begin{IEEEkeywords}
clock trees, useful skew, discrete optimization, infeasibility certificates,
logic-based Benders decomposition
\end{IEEEkeywords}

\section{Introduction}
Clock-skew scheduling trades arrival times against setup and hold margins
\cite{fishburn1990clock}. Prescribed-skew synthesis constructs a topology or
route for a schedule \cite{kourtev1999synthesis,tsao2000ust,held2003clock},
whereas useful-skew flows may replace exact targets with ranges or upper bounds
\cite{ewetz2017predictable,uysal2019latency}. A different question arises after
topology and tuning options are frozen: which schedules can the remaining
discrete choices realize exactly? Even a target inside the convex hull of
realizable schedules can be impossible.

We study this finite-model gap through logic-based Benders decomposition
(LBBD) \cite{hooker2003logic}. The scheduler proposes a latency vector; the
tree oracle returns a tuning witness or a checkable explanation and a valid
cut (Fig.~\ref{fig:loop}). Generic combinatorial Benders cuts and conflict
extraction already strengthen point no-goods
\cite{codato2006combinatorial,junker2004quickxplain}; their benefit must be
weighed against explanation cost \cite{cire2016computational}. Our
specialization is the combination of finite edge-domain proofs, laminar tree
contraction, and exact sink-projection feedback. Irreducible conflicts are not
necessarily minimum-cardinality certificates.

The contributions are: an exact relative-latency model with a sharp
two-sink complexity boundary; a sound and complete arithmetic certificate
system; and tree-specific projection cuts with an oracle-call separation from
no-goods. Same-model experiments isolate cut strength, projection compilation,
and end-to-end cost. \textbf{Fewer oracle calls do not imply faster closure};
the evaluation retains both resource limits and unfavorable comparisons.

\begin{figure}[t]
\centering
\includegraphics[width=\columnwidth]{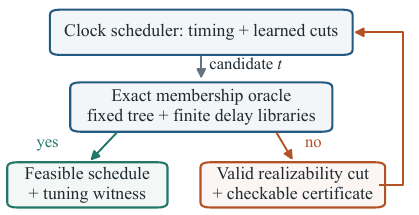}
\caption{Exact scheduler--tree feedback. Rejected targets generate sound
linear, congruence, or projected-relation cuts; a feasible target has a tuning
witness. All guarantees refer to the declared finite model.}
\label{fig:loop}
\end{figure}

\section{Exact Model and Tree Structure}
Let $T=(V,E)$ be a rooted tree with ordered sinks $L=(1,\ldots,n)$.
Each edge independently selects an integer delay $x_e\in D_e$ from a finite
nonempty library; a tunable node is folded into its incoming edge. With
root-to-sink incidence matrix $M$, absolute latency is $Mx$. Setup/hold
difference constraints ignore common latency, so define
\begin{equation}
 \Q(y)=y-y_1\mathbf1,\qquad
 \Real=\{\Q(Mx):x_e\in D_e\}.
 \label{eq:R}
\end{equation}
Targets are integer vectors $t\in\mathbb Z^n$ with $t_1=0$.
This is an additive, single-corner model, not a signoff timing model.
For selected sinks $C$, write $\pi_C(r)=\Q(r|_C)$.
With $a_e=\Q(M_{:,e})$, exact propagation starts from $S_0=\{0\}$ and adds
the finite set $\{d a_e:d\in D_e\}$ at each edge, merging duplicate states.
Induction yields the finite Minkowski sum $S_{|E|}=\Real$; predecessor records
recover a tuning witness. Intermediate state counts can be exponential.
Enumeration also finds the nearest vector: for $\delta=t-r$,
$\min_c\|t-r-c\mathbf1\|_\infty=(\max_i\delta_i-\min_i\delta_i)/2$.

\begin{theorem}
Exact membership is weakly NP-complete even for two sinks with binary
edge libraries.
\end{theorem}
\begin{proof}
An edge assignment is a polynomial witness. Give one sink fixed latency zero
and the other a chain with libraries $\{0,a_j\}$. Relative target $B$ is
realizable exactly when the positive integers $a_j$ have a subset summing to
$B$. Reachable-sum dynamic programming is pseudo-polynomial.
\end{proof}

For a star, membership simplifies to
$\bigcap_i(D_i-t_i)\ne\varnothing$. If each shifted set has at most $q$
elements, any empty intersection is witnessed by at most $q+1$ sets: choose
one set and, for each of its elements, a set omitting it. The bound is tight
for $U\setminus\{i\}$ with $|U|=q+1$. Thus pairwise checks are insufficient,
as $\{0,1\},\{0,2\},\{1,2\}$ demonstrate. Finding the smallest empty
subfamily is NP-hard by complementation from set cover.

\begin{theorem}[Laminar projection contraction]
For $|C|\ge2$, group edges by $\sigma_e=C\cap\mathrm{desc}(e)$, discard
$\sigma_e\in\{\varnothing,C\}$, and replace each remaining group by its
library sumset. The resulting model preserves $\pi_C(\Real)$ and has at most
$2|C|-2$ effective edge groups.
\label{thm:contraction}
\end{theorem}
\begin{proof}
All edges in one group have the same column $\Q(\mathbf1_{\sigma_e})$;
their Minkowski terms combine exactly by sumset addition. The discarded
columns are zero. Descendant sets form a laminar family on $C$, with at most
$2|C|-1$ nonempty members; removing $C$ gives the bound.
\end{proof}
\textbf{Structural projection size depends on the selected sinks}, not on
full-tree depth or edge count. The sumset cardinalities can nevertheless grow
exponentially; the theorem does not imply polynomial-time projection.

\section{Certificates and Scheduler Cuts}
\subsection{Checkable Unrealizability}
Choose sinks $L_C$ and constrained edges $E_C$. Restrict columns to $L_C$,
re-anchor, and delete the zero reference coordinate. Retain finite domains on
$E_C$ and relax every omitted relevant edge to any real delay. Let $A_R$
contain the relaxed columns. For each retained assignment
$z\in\prod_{e\in E_C}D_e$, the certificate stores
\begin{equation}
 \rho_z=\Q(t|_{L_C})-\sum_{e\in E_C}z_ea_e,\quad
 \alpha_z^T A_R=0,\quad \alpha_z^T\rho_z\ne0,
 \label{eq:trace}
\end{equation}
with the reference coordinate omitted and integer $\alpha_z$. Support is
$|L_C|+|E_C|$; trace length is a separate cost.

\begin{theorem}
The certificate system is sound and complete for unrealizability in
\eqref{eq:R}; exhaustive support search returns a global minimum.
\end{theorem}
\begin{proof}
Equation~\eqref{eq:trace} rules out every real completion of every retained
assignment, and relaxing domains can only enlarge the model. Conversely,
retain all relevant edges. For an unrealizable target every residual is
nonzero, $A_R$ is empty, and $\alpha_z=\rho_z$ is a proof. The support family
is finite, so exhaustive search with safe cardinality pruning finds a minimum.
\end{proof}
The constructor computes rational left-nullspace bases and clears
denominators. The independent verifier checks assignment coverage, recomputes
residuals, and checks both dot products. \textbf{The verifier certifies
infeasibility, not global minimality}; the latter requires exhaustive search.
Both search and explicit traces may be exponential.

\subsection{From Explanations to Realizability Cuts}
The master contains $\ell_{ij}\le t_i-t_j\le u_{ij}$ and objective $f(t)$.
For a translation-invariant linear form $\alpha$, independent edge choices give
the exact support bound
$\beta=\sum_e\max_{d\in D_e}d\sum_{i\text{ below }e}\alpha_i$.
If $\alpha^Tt>\beta$, the cut $\alpha^Ty\le\beta$ is valid. Congruence cuts
instead restrict $\alpha^Ty$ to its exact reachable residues modulo $m$.
Neither bounded arithmetic search is complete.

The complete fallback uses Theorem~\ref{thm:contraction} to compute the
projected relation $P_C=\pi_{L_C}(\Real)=\{r^1,\ldots,r^k\}$.
Its exact MILP encoding is
\begin{equation}
 \Q(t|_{L_C})=\sum_{j=1}^k\lambda_j r^j,\qquad
 \sum_{j=1}^k\lambda_j=1,\quad \lambda_j\in\{0,1\}.
 \label{eq:projection}
\end{equation}
This is a table constraint in a finite-domain master. All three cut families
preserve every vector in $\Real$; cuts are admitted only when they reject the
current target. Repeated separation therefore preserves the feasible optimum.
The implemented hierarchy tries support-two linear forms, congruences through
support four with moduli $2$--$12$, then linear forms through support four,
before minimum-certificate projection.

A certificate-core no-good rejects only the current projected value; exact
projection retains only reachable values on the same support. Their accepted
sets satisfy
\begin{equation}
 A_{\rm projection}\subseteq A_{\rm core}\subseteq A_{\rm point}.
 \label{eq:dominance}
\end{equation}
This is set dominance, not runtime dominance. Linear cuts cannot separate
holes inside $\conv(\Real)$: two libraries $\{0,2\}$ realize skews
$\{-2,0,2\}$ but not $1$. More strongly, let $\Real_N=\{0,2N\}$ and minimize
$|t-N|$ over integer $t\in[0,2N]$. Point and one-coordinate core no-goods
reject all $2N-1$ interior points before a feasible call, requiring $2N$ calls.
\textbf{Exact projection needs only two calls}; linear-only feedback stalls.

\begin{figure}[t]
\centering
\includegraphics[width=\columnwidth]{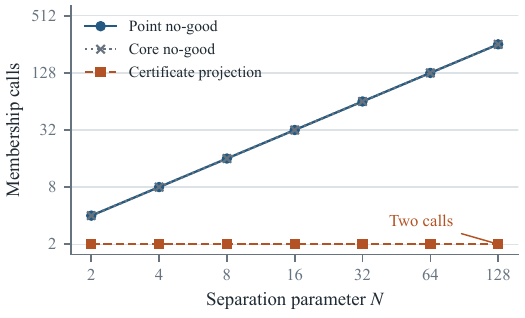}
\caption{Measured separation family on logarithmic axes. Point and core
no-goods coincide at $2N$ calls; exact projection uses two throughout.
This controlled family illustrates cut strength, not runtime speedup.}
\label{fig:separation}
\end{figure}

\begin{figure*}[t]
\centering
\includegraphics[width=\textwidth]{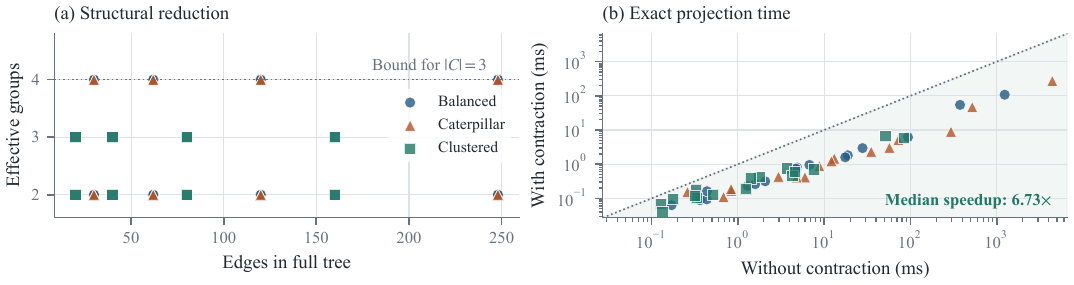}
\caption{Laminar contraction on 48 exact projections. (a) Full-tree edges
reduce to two to four effective groups; markers may overlap. (b) Repeated
vector-only projection time, including grouping and sumsets. Logarithmic axes;
the shaded region below the diagonal favors contraction. This is a
projection-stage ablation, not an end-to-end speedup.}
\label{fig:contraction}
\end{figure*}

\section{Evaluation}
\subsection{Common-Model Protocol}
The cut-isolation suite uses 160 distinct four-sink trees: four topologies,
four library patterns, and ten seeds. Each supplies a center hole and a
maximum-depth stress hole, yielding 320 workloads. Strategies share an
objective-ordered enumerative master and exact oracle; certificate strategies
share global minimum certificates. Every cut is checked against all reachable
vectors. All complete strategies attain the same optimum. Projection beats
point no-good in \textbf{295 workloads and ties in 25, with no loss}.
Table~\ref{tab:cuts} separates center and deepest-hole policies; linear-only
feedback stalls on 240/320 workloads. Payload counts stored coefficients,
residues, or projected coordinates, summed over cuts per workload. Full-domain
feedback achieves two calls but stores a much larger relation; the hierarchy
reduces payload while taking more calls than certificate projection.
The hierarchy emits 337 linear, 635 congruence, and 69 projection cuts;
the scalar-count proxy excludes solver-specific presolve.
Figure~\ref{fig:separation} independently checks the analytic call counts.
These controlled holes measure cut strength, not industrial prevalence.

\begin{table}[t]
\caption{Cut strength and size on 320 workloads. Center/deep columns:
mean calls on 160 workloads each. Payload: mean scalar count over all 320.
Dashes avoid conditional means for stalled strategies.}
\label{tab:cuts}
\centering\small
\setlength{\tabcolsep}{2.7pt}
\begin{tabular}{lrrrr}
\toprule
Feedback & Solved & Center & Deep & Payload\\
\midrule
Point no-good & 320/320 & 15.06 & 882.73 & 447.89\\
Certificate-core & 320/320 & 3.77 & 16.56 & 18.73\\
Linear only & 80/320 & -- & -- & --\\
Linear $\to$ proj. & 320/320 & 2.75 & 4.46 & 95.86\\
Cert. projection & 320/320 & \textbf{2.74} & \textbf{3.81} & 105.29\\
Compact hierarchy & 320/320 & 3.02 & 5.49 & 82.41\\
\midrule
Full-domain control & 320/320 & 2.00 & 2.00 & 880.45\\
\bottomrule
\end{tabular}

\end{table}

A separate suite uses 32 center-hole workloads with integer setup/hold-form
masters, anchored $L_1$ objectives, and SciPy/HiGHS. Every method receives one
warmup and five fresh-state runs in seeded shuffled order. Times include
relation construction, model building, cut validation, and solving; generation
and process startup are excluded. All 1,280 runs match independently
enumerated optima. Monolithic uses one-hot edge choices; table uses one-hot
reachable vectors; MDD uses flow through layered reachable states. These are
common-solver encodings, not implementations of the CP propagators in
\cite{andersen2007mdd,demeulenaere2016compact}. Deletion-MUS uses an
irreducible conflict; its subset checks count toward time.

\begin{table}[t]
\caption{Same-solver results on 32 masters. Times include preprocessing
and use five measured repetitions.}
\label{tab:solver}
\centering\small
\setlength{\tabcolsep}{3pt}
\begin{tabular}{lrrr}
\toprule
Method & Calls & Total ms & Peak binaries\\
\midrule
Point no-good & 2.97 & 19.28 & 11.8\\
Deletion-MUS core & 2.50 & 11.20 & 3.4\\
Minimum-cert. core & 2.31 & 11.18 & 2.9\\
Certificate projection & 2.28 & 9.15 & 13.2\\
Compact hierarchy & 2.31 & \textbf{7.61} & 3.2\\
\midrule
Monolithic exact & 1 solve & \textbf{2.46} & 14.2\\
Exact table & 1 solve & 13.35 & 194.1\\
Layered MDD & 1 solve & 18.78 & 336.2\\
\bottomrule
\end{tabular}

\end{table}

\subsection{Performance and Structural Ablation}
Table~\ref{tab:solver} shows that compact feedback is fastest among decomposed
methods, but \textbf{monolithic embedding is fastest overall}. Bold times
mark the best decomposed and direct methods. The interface
is therefore an alternative for separated scheduler/oracle architectures, not
a universal replacement for direct optimization.

The contraction ablation covers 48 projections on balanced, caterpillar, and
clustered trees with 16/32 sinks, selected support two/three, one/four edge
subdivisions, and binary/ternary libraries. Both paths compute vectors without
witnesses, with one warmup and five alternating measured repetitions.
\textbf{All 48 projected sets agree exactly}. Full trees with 20--248 edges
contract to 2--4 groups, with median projection-only speedup
\textbf{6.73$\times$} (range 1.68--32.89$\times$; Fig.~\ref{fig:contraction}).
Certificate-search cost is excluded from this isolated ablation.

All 320 initial targets have minimum certificates passing independent
verification: support ranges from four to eight and traces from 4 to 243
assignments. Table~\ref{tab:scaling} covers 24 general-tree cases through 32
sinks: monolithic and the hierarchy solve 24/24; minimum-certificate projection
solves 22/24 under a 10-s budget. All 12 explicit-DP cases at 16/32 sinks hit
the 250,000-state cap.
These libraries are all-even, allowing one congruence cut in every hierarchy
run; this favorable structure is not representative of all discrete holes.

\begin{table}[t]
\caption{General-tree scaling: six instances per size. Monolithic and
hierarchy solve every instance. $^*$Minimum-projection time includes only
completed instances; the 32-sink subset differs.}
\label{tab:scaling}
\centering\small
\setlength{\tabcolsep}{3pt}
\begin{tabular}{rrrrrr}
\toprule
 & \multicolumn{2}{c}{Completed / 6} & \multicolumn{3}{c}{Mean time (ms)}\\
\cmidrule(lr){2-3}\cmidrule(lr){4-6}
Sinks & DP & Min-proj. & Mono. & Min-proj.$^*$ & Hier.\\
\midrule
4 & 6/6 & 6/6 & 1.51 & 7.16 & 5.20\\
8 & 6/6 & 6/6 & 1.35 & 33.16 & 5.95\\
16 & 0/6 & 6/6 & 3.37 & 672.77 & 14.04\\
32 & 0/6 & \textbf{4/6} & 4.09 & 1816.28 & 75.00\\
\bottomrule
\end{tabular}

\end{table}

\subsection{Stress Beyond Parity}
We embed pair-gap and three-way obstructions into three tunable backgrounds
at 8/16/32 sinks. All terminal libraries cover every residue modulo
$2$--$12$, so no linear-form congruence with those moduli rejects the initial
target. Stored midpoint witnesses place the target inside the exact convex
hull; every full-tree sink pair in each three-way case is feasible. These
are controlled constructions, not independent physical samples.

\balance
All iterative methods share a bottom-up fixed-target membership oracle and
the same HiGHS master. Six methods run three times on 18 instances, each with
a 5-s budget after common backend warmup. Independent gadget lower bounds and
full-tree assignments certify the optimum; every successful run matches it.
All 138 timeouts remain in the 324-run denominator.

\begin{table}[t]
\caption{Non-parity stress: instances solved in all three 5-s runs.
Timeouts remain in each denominator.}
\label{tab:nonparity}
\centering\small
\setlength{\tabcolsep}{3pt}
\begin{tabular}{lrrr}
\toprule
Method & Pair / 9 & Triple / 9 & Total / 18\\
\midrule
Point no-good & 0 & 0 & 0\\
Deletion core & 9 & 0 & 9\\
Minimum-cert. projection & 9 & 4 & 13\\
Compact hierarchy & 3 & 3 & 6\\
Deletion-core projection & 9 & 7 & \textbf{16}\\
\midrule
Monolithic & 9 & 9 & \textbf{18}\\
\bottomrule
\end{tabular}

\end{table}

Minimum-certificate projection solves 13/18, versus 16/18 for projection on
cheaper deletion cores and 6/18 for the compact hierarchy
(Table~\ref{tab:nonparity}). Monolithic solves all 18. Successful triple-core
projections can contain thousands of tuples; deletion-projection's two
failures occur during master solving. \textbf{Minimum support and compact-cut
priority do not guarantee fast closure}. Solved subsets differ, so conditional
time averages are not used to rank these methods.

\subsection{Physical Scope}
Two OpenROAD case studies \cite{ajayi2019openroad} use GCD and AES trees with
35 and 562 original sinks. Each freezes placement and characterizes four
buffer sites over 81 choices, exposing discrete holes and support-four
certificates. On 100 adjacent-hole queries, nonlinear feedback and point
no-good both need two calls: \textbf{no call-count advantage is observed}.
Alternatives were not independently detailed-routed, and the GCD flow skips
optional timing repair. Coupling, slew, load dependence, and multi-corner
variation are outside the additive model; no signoff claim follows.

\section{Conclusion}
Arithmetic certificates and laminar projection contraction provide an exact,
checkable interface for discrete clock-tree realizability. Projection improves
elimination over same-support no-goods, while experiments delimit its practical
value: explanation cost, relation size, and unsuccessful compact-cut search can
dominate. Direct embedding remains the strongest tested end-to-end baseline.

\bibliographystyle{IEEEtran}
\bibliography{references}

@article{fishburn1990clock,
  author  = {John P. Fishburn},
  title   = {Clock Skew Optimization},
  journal = {IEEE Transactions on Computers},
  volume  = {39},
  number  = {7},
  pages   = {945--951},
  year    = {1990},
  doi     = {10.1109/12.55696}
}

@article{hooker2003logic,
  author  = {John N. Hooker and Greger Ottosson},
  title   = {Logic-Based Benders Decomposition},
  journal = {Mathematical Programming},
  volume  = {96},
  number  = {1},
  pages   = {33--60},
  year    = {2003},
  doi     = {10.1007/s10107-003-0375-9}
}

@article{codato2006combinatorial,
  author  = {Gianni Codato and Matteo Fischetti},
  title   = {Combinatorial Benders' Cuts for Mixed-Integer Linear Programming},
  journal = {Operations Research},
  volume  = {54},
  number  = {4},
  pages   = {756--766},
  year    = {2006},
  doi     = {10.1287/opre.1060.0286}
}

@article{cire2016computational,
  author  = {Andr\'{e} A. Cir\'{e} and Elvin \c{C}oban and John N. Hooker},
  title   = {Logic-Based Benders Decomposition for Planning and Scheduling: A Computational Analysis},
  journal = {The Knowledge Engineering Review},
  volume  = {31},
  number  = {5},
  pages   = {440--451},
  year    = {2016},
  doi     = {10.1017/S0269888916000254}
}

@article{kourtev1999synthesis,
  author  = {Ivan S. Kourtev and Eby G. Friedman},
  title   = {Synthesis of Clock Tree Topologies to Implement Nonzero Clock Skew Schedule},
  journal = {IEE Proceedings -- Circuits, Devices and Systems},
  volume  = {146},
  number  = {6},
  pages   = {321--326},
  year    = {1999},
  doi     = {10.1049/ip-cds:19990582}
}

@inproceedings{tsao2000ust,
  author    = {Chung-Wen Albert Tsao and Cheng-Kok Koh},
  title     = {{UST/DME}: A Clock Tree Router for General Skew Constraints},
  booktitle = {Proceedings of the IEEE/ACM International Conference on Computer-Aided Design},
  pages     = {400--405},
  year      = {2000},
  doi       = {10.1109/ICCAD.2000.896505}
}

@inproceedings{held2003clock,
  author    = {Stephan Held and Bernhard Korte and Jens Massberg and Matthias Ringe and Jens Vygen},
  title     = {Clock Scheduling and Clocktree Construction for High Performance {ASICs}},
  booktitle = {Proceedings of the IEEE/ACM International Conference on Computer-Aided Design},
  pages     = {232--240},
  year      = {2003},
  doi       = {10.1109/ICCAD.2003.48}
}

@inproceedings{ewetz2017predictable,
  author    = {Rickard Ewetz},
  title     = {A Clock Tree Optimization Framework with Predictable Timing Quality},
  booktitle = {Proceedings of the Design Automation Conference},
  pages     = {72:1--72:6},
  year      = {2017},
  doi       = {10.1145/3061639.3062184}
}

@inproceedings{uysal2019latency,
  author    = {Necati Uysal and Wen-Hao Liu and Rickard Ewetz},
  title     = {Latency Constraint Guided Buffer Sizing and Layer Assignment for Clock Trees With Useful Skew},
  booktitle = {Proceedings of the Asia and South Pacific Design Automation Conference},
  pages     = {761--766},
  year      = {2019},
  doi       = {10.1145/3287624.3287681}
}

@inproceedings{ajayi2019openroad,
  author    = {Tutu Ajayi and Vidya A. Chhabria and Mateus Foga\c{c}a and Soheil Hashemi and Abdelrahman Hosny and Andrew B. Kahng and Minsoo Kim and Jeongsup Lee and Uday Mallappa and Marina Neseem and Geraldo Pradipta and Sherief Reda and Mehdi Saligane and Sachin S. Sapatnekar and Carl Sechen and Mohamed Shalan and William Swartz and Lutong Wang and Zhehong Wang and Mingyu Woo and Bangqi Xu},
  title     = {Toward an Open-Source Digital Flow: First Learnings from the {OpenROAD} Project},
  booktitle = {Proceedings of the Design Automation Conference},
  pages     = {76:1--76:4},
  year      = {2019},
  doi       = {10.1145/3316781.3326334}
}

@inproceedings{junker2004quickxplain,
  author    = {Ulrich Junker},
  title     = {{QUICKXPLAIN}: Preferred Explanations and Relaxations for Over-Constrained Problems},
  booktitle = {Proceedings of the AAAI Conference on Artificial Intelligence},
  pages     = {167--172},
  year      = {2004}
}

@inproceedings{andersen2007mdd,
  author    = {Henrik Reif Andersen and Tarik Hadzic and John N. Hooker and Peter Tiedemann},
  title     = {A Constraint Store Based on Multivalued Decision Diagrams},
  booktitle = {Principles and Practice of Constraint Programming},
  pages     = {118--132},
  year      = {2007},
  doi       = {10.1007/978-3-540-74970-7_11}
}

@inproceedings{demeulenaere2016compact,
  author    = {Jordan Demeulenaere and Renaud Hartert and Christophe Lecoutre and Guillaume Perez and Laurent Perron and Jean-Charles R\'{e}gin and Pierre Schaus},
  title     = {Compact-Table: Efficiently Filtering Table Constraints with Reversible Sparse Bit-Sets},
  booktitle = {Principles and Practice of Constraint Programming},
  pages     = {207--223},
  year      = {2016},
  doi       = {10.1007/978-3-319-44953-1_14}
}
\end{document}